\documentclass[journal,letter]{IEEEtran}

\usepackage{microtype}
\usepackage{capt-of}
\usepackage{siunitx}
\usepackage{graphicx}
\usepackage{amssymb}
\usepackage{amsmath}
\usepackage{algorithm}
\usepackage{algpseudocode}
\usepackage{cite}
\usepackage{stfloats}
\usepackage{epstopdf}
\usepackage{mdwlist}
\usepackage{threeparttable}
\usepackage[all,2cell]{xy} \UseAllTwocells
\usepackage{booktabs}
\usepackage{fancyvrb}
\usepackage{balance}
\usepackage{tensor}
\usepackage{leftidx}
\usepackage{mathabx}
\usepackage{color}
\usepackage{csquotes}
\usepackage{pdflscape}
\usepackage{afterpage}
\usepackage{array}
\usepackage{cases}
\usepackage{bm}
  \usepackage[caption=false,font=footnotesize]{subfig}
  \usepackage{cleveref}

\newtheorem{theorem}{Theorem}

\newtheorem{corollary}{Corollary}

\newtheorem{remark}{Remark}

\AtBeginDocument{}

\newcommand{\Pe}{P_\mathrm{b}}
\newcommand{\erfc}{\mathrm{erfc}}

\newcommand{\diff}{\mathrm{d}}

\newcommand{\kf}{\kappa_f}
\newcommand{\kb}{\kappa_b}

\newcommand{\uu}{\mathbf{u}}

\newcommand{\fscale}{0.48}

\newcommand{\A}{\mathrm{A}}
\newcommand{\B}{\mathrm{B}}
\newcommand{\C}{\mathrm{C}}
\newcommand{\dt}{\Delta t}

\newcommand{\Tx}{\mathrm{Tx}}
\newcommand{\Rx}{\mathrm{Rx}}

\newcounter{eqnback1}

\begin{document}
	
	\title{\huge Chemical Reactions-based Detection Mechanism for Molecular Communications}
	\author{\IEEEauthorblockN{Trang Ngoc Cao\IEEEauthorrefmark{1}, Vahid Jamali\IEEEauthorrefmark{2}, Wayan Wicke\IEEEauthorrefmark{3},  Nikola Zlatanov\IEEEauthorrefmark{4}, \\ Phee Lep Yeoh\IEEEauthorrefmark{5}, Jamie Evans\IEEEauthorrefmark{6}, and Robert Schober\IEEEauthorrefmark{3}}\\
		

\thanks{This paper was presented in part at the
	IEEE Wireless Communications and Networking Conference 2020 \cite{Tra:20:WCNC}.}%

\thanks{T. N. Cao is with the School of Psychological Sciences, Monash University, Melbourne, VIC 3800, Australia
	(e-mail: trang.cao@monash.edu).}
\thanks{ V. Jamali is
	with the Department of Electrical and Computer,
	Princeton University, New Jersey 08544, Germany (e-mail:
	jamali@princeton.edu).}
\thanks{ W. Wicke and R. Schober are
	with the Institute for Digital Communications, Friedrich-Alexander-
	Universit\"at Erlangen-N\"urnberg, Erlangen 91058, Germany (e-mail:
	 wayan.wicke@fau.de; robert.schober@fau.de).}
\thanks{N. Zlatanov is  with the Department of Electrical and Computer Systems
	Engineering, Monash University, Melbourne, VIC 3800, Australia (e-mail:
	nikola.zlatanov@monash.edu).}
\thanks{P. L. Yeoh is with the School of Electrical and Information Engineering, University of Sydney, Sydney, NSW 2006, Australia
	(e-mail: phee.yeoh@sydney.edu.au).}
\thanks{J. S. Evans is with the Department of Electrical and Electronic
	Engineering, University of Melbourne, Melbourne, VIC 3010, Australia
	(e-mail: jse@unimelb.edu.au).}
		}%

	\maketitle
	\begin{abstract}

In molecular communications, the direct detection of signaling molecules may be challenging due to a lack of suitable sensors and interference from co-existing substances in the environment. Motivated by research in molecular biology, we investigate an indirect detection mechanism  using chemical reactions between the signaling molecules and a molecular probe to produce an easy-to-measure product  at the receiver. 
We consider two implementations of the proposed detection mechanism, i.e.,  unrestricted probe movement and  probes restricted to a volume around the receiver.
In general, the resulting reaction-diffusion equations that describe the concentrations of the reactant and product molecules in the system are non-linear and coupled, and cannot be solved in closed form. To evaluate these molecule concentrations, we develop an efficient iterative algorithm by discretizing the time variable and solving for the space variables of the concentration equations in each time step. The accuracy of our proposed algorithm is verified by particle-based simulations.  
 In the special case when the concentration of the unrestricted probes is high and not changed significantly by the chemical reaction compared to the signaling molecule concentration, we can obtain insightful closed-form solutions. 
 Our results show that the  concentrations of the product molecules and the signalling molecules share a similar characteristic over time, i.e., a single peak and a long tail. The peak and tail values of the product molecule concentration can be controlled by choosing  probes with suitable parameters, e.g., the diffusion coefficient,  reaction rate, and  released quantity. 
 %
 We analyze the bit error rate (BER) of the system for a threshold decision rule. Furthermore, we highlight that by carefully choosing the molecular probe and  optimizing  the  decision  threshold, the BER  can be improved significantly
   and outperform that of a direct detection system. Moreover, when  molecular probes are kept in a small volume around the receiver, fewer resources, i.e.,  probe molecules, are needed to achieve the same BER and an even higher data rate compared to the case when they are not restricted.

	\end{abstract}
	\section{Introduction}\label{sec:1}

	In molecular communications (MC), information is typically encoded in the number, type, or time of the release of signaling molecules.   The encoded information is detected at the receiver by a sensor \cite{JAW:18:Arxiv,Bi:22:Sur,Pan:22:NM}. Therefore, sensor technology, in particular  chemical sensors, plays an important role for the design of receivers in MC systems. 
	
	Chemical sensors are designed to provide a measurable signal corresponding to the concentration of an analyte (i.e.,  a chemical substance) in the  environment \cite{JCW:19:che}. This measurement can be based on   magnetic or electrical fields, resistance, capacitance, inductance, or an optical response \cite{JCW:19:che}.  In MC, the selection of the sensor technique depends on the specific requirements of the considered application. For example, magnetic field based sensing was used in \cite{Way:21:TMBMC} and resistance based sensing was applied in \cite{FGE:13:PO}. The systems in \cite{Way:21:TMBMC} and \cite{FGE:13:PO} have demonstrated
	the possibility of realizing MC but they are fairly simple since there are no interfering sources impairing  the detection of the
	signaling molecules, i.e., no other magnetic  \cite{Way:21:TMBMC} or alcohol sources \cite{FGE:13:PO} besides the desired signal. Nevertheless,
	in many practical applications of MC, e.g.,  drug delivery and health monitoring, there usually exist  other chemical substances which may cause interference for the detection of the signaling molecules. Environmental monitoring applications also need to handle environments where many different chemicals and electromagnetic sources are present and potentially cause interference. For example,  chemicals such as zinc and copper have similar magnetic susceptibility and electrical resistivity and thus are difficult to distinguish at the receiver.  In such cases, one possible solution for detection is to employ unique chemical reactions where only the signaling molecule, i.e., the analyte, can react with a specific reactant, i.e., a molecular probe, to produce a product molecule which can be easily measured.  This approach has been an area of intense research in molecular biology, see \cite{JCW:19:che} and references therein. For example, zinc ions react with spiropyran and produce a merocyanine metal complex, which exhibits florescence, i.e., it emits light of a wavelength that can be measured via optical spectroscopy \cite{Nat:10:Tet,JCW:19:che}. 
	Furthermore, synthesizing  molecular probes that are matched to a given analyte and the considered environment has been an active area of research, see \cite{JCW:19:che} and references therein.
	
	Note that, in some MC systems, the signaling molecules should be small and lightweight, e.g., zinc ions or calcium ions, such that they can be easily stored at the transmitter and can  diffuse quickly from the transmitter to the receiver. On the other hand,  the product molecules of the reaction, i.e., the combination of the probe and the analyte \cite{LYY:18:CO}, which can be detected directly by the receiver,  are usually larger molecules and thus may not be   suitable as quickly-diffusive signaling molecules. 
	Moreover,  when the reaction occurs, a measurable signal, e.g., light, corresponding to the reaction product may be generated but then  disappear quickly by a process referred to as quenching \cite{JCW:19:che,Zha:16:MC, Liu:17:Lum}, which is useful for reducing inter-symbol interference (ISI). 
	Motivated by these advantages, in this work, we propose a novel MC detection  mechanism  based on the reaction of signaling molecules with a molecular probe. 
	
		Chemical reactions have been studied in different contexts for  MC. For example, chemical reactions were used to generate signaling molecules at the transmitter \cite{Yan:20:TCOM} and  potent drugs  on the surface of the receiver \cite{CMM:16:NB}. The reactions of signalling molecules with enzymes in the environment  were exploited  to mitigate ISI in \cite{NCS:14:INB,Chan:17:ETT}.   The reactions of signalling molecules with molecules  around the  receiver were used for the detection in \cite{Adam:19:TCOM,Cho:13:NB,CMM:16:NB,Yuting:22:TCOM}. Chemical reactions have also been considered for coding and modulation in \cite{FAM:19:TCOM,JFS:19:MBSC, Abin:21:TCOM}. In  \cite{Yan:20:TCOM}, the chemical reactions were assumed to occur in a one dimensional environment and the concentration of one reactant was known.  The enzyme in \cite{CMM:16:NB} and the receptors in \cite{Adam:19:TCOM,Cho:13:NB,CMM:16:NB,Yuting:22:TCOM}, i.e., one of the reactants,  were assumed to be immobile. Moreover, the authors in \cite{NCS:14:INB,Chan:17:ETT} considered a fast reaction where the concentration of the enzymes remained constant. The authors in \cite{FAM:19:TCOM,JFS:19:MBSC} focused on the concentration of the signaling molecules, i.e., the reactants, but the products of the reaction were of no interest and not studied. In \cite{ZDN:18:Glo},  the  molecules emitted by the transmitter and the product of the reaction  during propagation were both considered, but the reaction was a degradation reaction and thus modeled as a first-order reaction. In \cite{Abin:21:TCOM},  approximate solutions of the reaction-diffusion equations are given in  the  forms  of    infinite  series  of  functions and recursive  equations are used to solve for these functions. However, the approximate solutions is shown to only converge to  the  true  solution  if  either  the  simulation  time  interval  or the reaction  rate is sufficiently small. Moreover, \cite{Abin:21:TCOM} consider the reaction-diffusion equations for modulation in a two-transmitter system without ISI.
		In this work, we consider the reaction between signaling molecules and  molecular probes, which has to be modeled as a  second-order reaction. To the best of the authors' knowledge,  second-order reactions with the reactants not being bound to the receiver but diffusing in the environment  have not previously been considered for detection design of MC systems with ISI. 
		
		We consider two possible implementations of the proposed detection mechanism. We first study the general and simple implementation where the molecular probes are released and then freely diffuse in an unbounded environment. To analyze the system performance, we propose an algorithm to determine the concentration of the product of the detection reaction.  We then investigate the special case  where the concentration of the molecular probes is high compared to the concentration of the analyte and thus not significantly affected by the chemical reaction. In this case, a closed-form solution for the concentration of the reaction product can be obtained.   For the second implementation, instead of distributing the molecular probes everywhere in the environment, we aim for an efficient usage of the probes. To this end, the probes are trapped in a small volume around the receiver. By doing this, the probes can react with the signaling molecules and create more products to be detected by the receiver instead of dispersing in the environment and  being wasted. The probes can be kept inside a volume by a natural or synthetic membrane that only allows the signaling molecules to pass through. In biological systems, a semipermeable  natural membrane can block certain types of molecules and allow others to diffuse across it via a process called osmosis \cite{ZAC:19:book, AJL:15:Book}. Synthetic membranes can also separate the probes from the external environment via a membrane process, which can be driven by pressure, concentration, or an electric field across the membrane \cite{Li:07:Book}. To analyze the system performance for this implementation, we suitably modify the algorithm developed for freely diffusing probes.

In this paper, we propose a novel MC detection concept  based on a chemical reaction between diffusive molecules and make the following main contributions:
	\begin{itemize}
		\item We propose a novel detection mechanism for MC systems in which the direct detection of the signaling molecules is not possible or not efficient. A molecular probe is employed to convert the original signaling molecules into product molecules which can be efficiently detected.
		
		\item We develop a robust iterative algorithm for evaluation of the spatio-temporal distribution of the product molecules by solving the underlying non-linear and coupled reaction diffusion equations. We apply the algorithm to two system implementations with suitable adaptations.
		
		\item We derive a closed-form expression for the concentration of the product molecules for the  case where the concentration of the unrestricted molecular probe is not significantly affected by the chemical reaction.
		
		\item We design a system where the probes are restricted to a volume around the receiver to enable an efficient use of the available resources.

		\item We analyze the performance of the  proposed detection mechanism  in terms of the bit error rate (BER). Furthermore, we provide new insights for  system design with regard to the optimal decision threshold, data rate, efficiency, and the  molecular probe's parameters such as the diffusion coefficient,  reaction rate, and  released quantity.

	\end{itemize}

This paper is an extension of a conference paper \cite{Tra:20:WCNC} which did not consider restricted  molecular probe movement and the resulting performance in terms of resource efficiency, BER, and data rate. Moreover, the impact of the system parameters on the product molecule concentration is  analyzed more in depth in this paper compared to \cite{Tra:20:WCNC}.
	
The remainder of this paper is organized as follows. In Section II, we introduce the system model and the proposed detection mechanism. In Section III, we analyze the molecule concentrations for three implementations of the proposed detection mechanism  and the resulting detection performance.  Numerical results are presented in Section IV, and
	Section V concludes the paper.

	\begin{figure}[!t]
		\centering
		\includegraphics[width=0.51\textwidth,trim=60 230 40 270,clip]{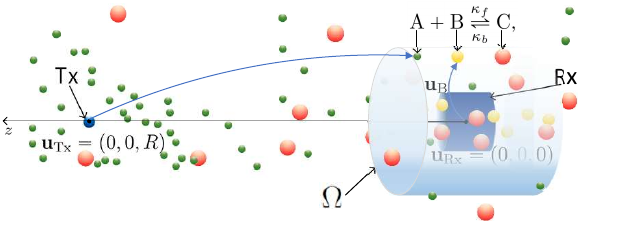}
		\caption{
			Schematic illustration of the system model. Type $\A$ molecules are released from the transmitter, $\Tx$, and react with type $\B$ molecules, released at position $\uu_\B$, in order to create type $\C$ molecules, which can be measured by the receiver, $\Rx$. The $\B$ molecules can be restricted in  volume $\Omega$ around $\Rx$.}
		\label{fig:3}
	\end{figure}

	\section{System Model and Detection Mechanism}\label{sec:2}
	We consider an MC system consisting of a point source transmitter, denoted by $\Tx$, and a transparent receiver, denoted by $\Rx$, in an unbounded three dimensional (3D) diffusive environment
with constant temperature and viscosity. The receiver has volume $\mathcal{V}_{\Rx}$ and is located at distance $R$ from the transmitter. Using  cylindrical coordinates \footnote{We choose cylindrical coordinates so that calculations for the symmetric systems can be simplified.}, where  position $\uu$ is defined as $\uu=(\rho,\phi,z), \rho\in[0,\infty), \phi\in[0,2\pi)$, and $z=(-\infty,\infty)$,  $\Tx$ and $\Rx$ are centered at $\uu_\Tx=(0,0,R)$ and $\uu_\Rx=(0,0,0)$, respectively. 
Let $T$ denote the duration of a symbol interval.
	We assume on-off keying modulation and that $\Tx$ releases  $N_\A$   molecules of type $\A$ to convey bit $``1"$ and  no molecules to convey bit $``0"$ at the beginning of the symbol interval, i.e., at $t =nT, n=0,1,\dots, N$, where $N$ is the length of the bit sequence. We assume that bits $``0"$ and $``1"$ have equal probabilities.  

%
%

\setcounter{eqnback1}{\value{equation}} \setcounter{equation}{1}
\begin{figure*}[!b]
		\hrulefill
	\begin{IEEEeqnarray}{lll} \label{eq:27}
		\frac{\partial C_\A(\uu,t)}{\partial t}=G_\A(\uu,t)+D_\A \nabla^2C_\A(\uu,t) -\kf C_\A(\uu,t) C_\B(\uu,t) +\kb C_\C(\uu,t),\IEEEyesnumber\IEEEyessubnumber \label{eq:27a}\\
		\frac{\partial C_\B(\uu,t)}{\partial t}=G_\B(\uu,t)+D_\B \nabla^2C_\B(\uu,t)-\kf C_\A(\uu,t) C_\B(\uu,t) +\kb C_\C(\uu,t),\IEEEyessubnumber\label{eq:27b}\\
		\frac{\partial C_\C(\uu,t)}{\partial t}=D_\C \nabla^2C_\C(\uu,t)+\kf C_\A(\uu,t) C_\B(\uu,t) -\kb C_\C(\uu,t), \IEEEyessubnumber \label{eq:27c}
	\end{IEEEeqnarray}

\end{figure*}

	We assume that the type $\A$ molecules cannot be detected directly  at the receiver as a suitable sensor is not available. Hence,   type $\B$ molecules are introduced into the system to react with the type $\A$ molecules  to create type $\C$ molecules  for which suitable sensors are available, see Fig.~\ref{fig:3}.  The type $\B$ molecules are referred to as  molecular probes \cite{Her:08:Book}. The type $\B$ molecules may be released at a fixed position, denoted by $\uu_\B$, e.g., $\uu_\B=\uu_\Rx$, or uniformly throughout the environment. Furthermore, the type $\B$ molecules may be restricted to a volume around the receiver, denoted by $\Omega$.  In cylindrical coordinates, we assume $\Omega$ to be the cylindrical volume bounded by $0\leq \rho\leq a$, $0\leq \phi <2\pi$, and $-\frac{b}{2}\leq z \leq \frac{b}{2}$.\footnote{ Cylinder, i.e., rod shape, is a common morphology of cells and bacteria, which was optimized by evolution  \cite{Vin:03:EMBO,You:07:COM}.  Since MC is provisioned to be applied in biological systems such as human bodies, cylinder is a good choice for MC system designs in order to leverage the advantages of this morphology.} When type $\B$ molecules diffuse freely in an unbounded environment, we have $a, b \rightarrow \infty$.  We model the sensing process via a transparent receiver which counts the number of the type $\C$ molecules in its volume without affecting the molecules. For example, for the detection of the zinc ions mentioned in the introduction, optical spectroscopy is used to measure the light intensity which is proportional to the number of product molecules, i.e., the type $\C$ molecules. We note that the effect of quenching, which could be exploited for ISI reduction, is neglected and  left for future work. We assume that the chemical reaction between the type $\A$ and $\B$ molecules is reversible and can be modeled as follows
	\setcounter{equation}{0}
	\begin{align} \label{eq:1}
	\A+\B\underset{\kb}{\overset{\kf}{\rightleftharpoons} }\C,
	\end{align}
	where $\kf$ is the forward reaction rate constant of a second-order reaction and $\kb$ is the backward reaction rate constant of a first-order reaction.
	We assume that the type $\A$, $\B$, and $\C$ molecules diffuse in the unbounded 3D  environment with diffusion coefficients $D_\A$, $D_\B$, and $D_\C$, respectively. Thereby, the concentrations of the type $\A$, $\B$, and $\C$ molecules at time $t$ and position $\uu$, denoted by $C_i(\uu,t), i\in\left\{\A,\B,\C \right\}$,  are governed by a set of reaction diffusion  equations in \eqref{eq:27} \cite[Eqs.~(9.5), (9.13)]{Cha:05:Book}, given at the bottom of this page,
	where  $\nabla^2$ is the Laplace operator and  
	\setcounter{equation}{2}
	\begin{align}
	\label{eq:28}
	G_i(\uu,t)=\sum_{t_i} \sum_{\uu_i}	N_i\delta_{\mathrm{d}}(t-t_i)\delta_{\mathrm{d}}(\uu-\uu_i), \hspace{1cm} i \in\left\{\A,\B\right\},
	\end{align}
	represents the concentration of the type $i$ molecules that are released into the channel.  In \eqref{eq:28},  $\delta_{\mathrm{d}}(\cdot)$, $N_i$,  $t_i$, and $\uu_i$  are the Dirac delta function,  the  number of molecules  released at time $t_i$,  the release times, and  the release positions of the type $i$ molecules, respectively. 
	The partial differential equations (PDEs) in \eqref{eq:27} are non-linear and coupled, i.e., the concentration of the type $i$ molecules after $n$ releases is not equal to the sum of the concentrations originating from each release. 
	Thus, the PDEs in \eqref{eq:27} do not have a closed-form solution in general \cite{Debnath:11:Book,NCS:14:INB,JFS:19:MBSC}.
	On the other hand, we need to obtain  $C_\C(\uu,t)$  to design the proposed MC system and to analyze its performance. To this end, in the next section, we will propose an algorithm for solving \eqref{eq:27} for  $C_\C(\uu,t)$ and analyze the system performance.

	\section{System  Analysis Framework} \label{sec:3}
	In this section, we consider  two implementations of the proposed detection mechanism. We first start with an implementation of the detector when the location of the probes is not restricted. We then consider the special case where the  concentration of the unrestricted probes is constant over time. Second, we investigate an implementation of the detector with probes restricted in a small volume around the receiver. We present an efficient numerical algorithm for determining $C_i(\uu,t)$, $i\in\left\{\A,\B,\C \right\}$, which can be used for both implementations.  For the special case of  unrestricted probes of constant concentration, we derive  analytical expressions for $C_i(\uu,t)$. Finally, we analyze the  performance of the system for a simple decision rule.
	\subsection{Detection with Unrestricted Probes}\label{sub3:1}

	 In this subsection, we consider the detection in an unbounded environment, i.e.,  bounds $a, b \rightarrow \infty$ in cylindrical coordinates. For determining $C_i(\uu,t)$, we adapt   \cite[Algorithm~1]{JFS:19:MBSC} to the problem at hand. The basic concept behind this algorithm is to discretize the time variable and  solve the PDEs in terms of the space variable. Considering small time intervals allows us to decouple the diffusion and reaction equations\footnote{For a detailed mathematical proof of the accuracy of the algorithm with respect to the decoupling of diffusion and reaction, please refer to \cite{JFS:19:MBSC}.}. In particular, Algorithm~\ref{al:1}, shown at the top of the page, summarizes the  steps needed for calculating the concentrations of the type $\A, \B$, and $\C$ molecules, where $T^{\max}$ is the maximum  time considered. We will verify the accuracy of the resulting numerical algorithm via particle-based simulation in Section~\ref{sec:4}. 
	 
	 \begin{algorithm}[t!]
	 	
	 	\caption{Interative Calculation of the Concentrations}
	 	
	 	\begin{algorithmic}
	 		\State\textit{Initialization}: $t=0, \dt, T^{\max}$, and $C_i(\uu,t=0)$.
	 		\While {{$t\leq T^{\max}$}}
	 		\State Update $t$ as $t+\dt$.
	 		\State Compute $\bar{G}_i(\uu)$  by \eqref{eq:23}, $C_i^\mathrm{df}(\uu,t)$ by \eqref{eq:25}, and $C_i^\mathrm{rc}(\uu,t)$ by \eqref{eq:5}, concurrently for all $i$.
	 		\State Update $C_i(\uu,t)$ based on \eqref{eq:22}, concurrently for all $i$.
	 		\EndWhile
	 		\State Return $C_i(\uu,t)$.
	 	\end{algorithmic}
	 	\label{al:1}
	 \end{algorithm}
	 
	  In Algorithm~\ref{al:1}, $C_i(\uu,t)$ is updated for  $i=\left\{\A,\B,\C\right\}$ in each iteration as follows \cite{JFS:19:MBSC}
	\setcounter{equation}{3}
	\begin{align}
	\label{eq:22}
	C_i(\uu,t+\dt)&=\bar{G}_i(\uu)+C_i^\mathrm{df}(\uu,t+\dt)\nonumber\\
	&\quad+C_i^\mathrm{rc}(\uu,t+\dt)-C_i(\uu,t),
	\end{align}
	where $\bar{G}_i(\uu)$ is the concentration of the type $i$ molecules  released at $\uu$ in  time interval $\left[t,t+\Delta t\right]$. $C_i^\mathrm{df}(\uu,t+\dt)$ and $C_i^\mathrm{rc}(\uu,t+\dt)$ are the concentrations of the type $i$ molecules assuming that in interval $[t,t+\Delta t]$ only diffusion and only reactions occur, respectively, while the other phenomenon is absent. The  updates  of the concentrations in \eqref{eq:22} are given in the following.
	
	\setcounter{eqnback1}{\value{equation}} \setcounter{equation}{8}
	\begin{figure*}[!b]
		
		\hrulefill
		\begin{IEEEeqnarray}{lll} \label{eq:5}
			C_\A^\mathrm{rc}(\uu,t+\Delta t)=\frac{c_2(\uu)+\kf c_{11}(\uu)-\kb-\left(c_2(\uu)-\kf c_{11}(\uu)+\kb\right)c_4(\uu) \exp\left(-c_2(\uu)\Delta t\right)}{2\kf \left(1+c_4(\uu) \exp\left(-c_2(\uu) \Delta t\right)\right)},\nonumber\\\IEEEyesnumber\IEEEyessubnumber \label{eq:5a}\\
			C_\B^\mathrm{rc}(\uu,t+\Delta t)=\frac{c_2(\uu)-\kf c_{11}(\uu)-\kb-\left(c_2(\uu)+\kf c_{11}(\uu)+\kb\right)c_4(\uu) \exp\left(-c_2(\uu)\Delta t\right)}{2\kf \left(1+c_4(\uu) \exp\left(-c_2(\uu) \Delta t\right)\right)},\nonumber\\\IEEEyessubnumber\label{eq:5b}\\
			C_\C^\mathrm{rc}(\uu,t+\Delta t)=c_{12}(\uu)-C_\A^\mathrm{rc}(\uu,t+\Delta t), \IEEEyessubnumber \label{eq:5c}
		\end{IEEEeqnarray} 

	\end{figure*}
	
	\subsubsection{	Update of Release}
	As proved in \cite{JFS:19:MBSC}, $\bar{G}_i(\uu)$ is given by
	\setcounter{equation}{4}
	\begin{align}
	\label{eq:23}
	\bar{G}_i(\uu)=\sum_{t_i}N_i\delta_{\mathrm{d}}(t+\dt-\varepsilon -t_i)\delta_{\mathrm{d}}(\uu),
	\end{align}
	where $\varepsilon$ is an arbitrary small positive real number satisfying $\varepsilon \ll \Delta t$ and $t_i$ is given by $t_i=m\Delta t -\varepsilon, m \in \mathbb{N}$.
	\subsubsection{Update of $C_i^\mathrm{df}(\uu,t+\dt)$}
	As shown in \cite{JFS:19:MBSC},
	\begin{align}
	\label{eq:25}
	C_i^\mathrm{df}(\uu,t+\dt)&=\frac{1}{\left(4\pi D_i \dt\right)^{3/2}}\int_{\tilde{\uu}} C_i(\tilde{\uu},t)\nonumber\\
	&\quad\times\exp\left(-\frac{||\uu-\tilde{\uu}||^2}{4D_i\dt}\right)\diff \tilde{\uu}.
	\end{align}

	 Due to the symmetry of the system, we  choose cylinder coordinates to simplify the calculation of \eqref{eq:25}   in Corollary~\ref{cor:1}.
	 \begin{corollary}\label{cor:1}
	 	Using cylindrical coordinates,  	$C_i^\mathrm{df}(\uu,t+\dt)$ is  given by
	 	\begin{align}\label{eq:35}
	 		C_i^\mathrm{df}(\uu,t+\dt)&=\frac{2\pi}{\left(4\pi D_i \dt\right)^{3/2}}\int_{\tilde{\rho}=0}^\infty\int_{\tilde{z}=-\infty}^\infty C_i(\tilde{\rho},\tilde{z},t)\nonumber\\
	 		&\quad \times W_i^z(z,\tilde{z})W_i^\rho(\rho,\tilde{\rho}) ~\diff \tilde{z}\diff \tilde{\rho},
	 	\end{align}
	 	where
	 						\begin{IEEEeqnarray}{lll} \label{eq:36a}
	 						W_i^z(z,\tilde{z})=\exp\left(-\frac{\left(z-\tilde{z}\right)^2}{4D_i \Delta t}\right),\IEEEyesnumber\IEEEyessubnumber \label{eq:36}\\
	 							W_i^\rho(\rho,\tilde{\rho})=\exp\left(-\frac{\rho^2-\tilde{\rho}^2}{4D_i \Delta t}\right)\tilde{\rho}~I_0\left(\frac{\rho\tilde{\rho}}{2 D_i \Delta t}\right), \IEEEyessubnumber \label{eq:37}
	 						\end{IEEEeqnarray} 
	 	and $I_0(\cdot)$ is the zeroth order modified Bessel function of the first kind.
	 \end{corollary}	
	 \begin{proof}
	 	Please refer to Appendix~\ref{app:1}.
	 	\end{proof}
Note that $W_i^z(z,\tilde{z})$ and $W_i^\rho(\rho,\tilde{\rho})$ do not change over time, and thus, can be evaluated offline and used online in order to reduce computational complexity. 
	
	\subsubsection{Update of $C_i^\mathrm{rc}(\uu,t)$}
	Since, in this work,  the product of the reaction  is used for detection whereas in \cite{JFS:19:MBSC} it was of no interest for the considered system, the reaction diffusion  equations in \cite{JFS:19:MBSC} are different from those in this work. Hence,  in order to use Algorithm~\ref{al:1}, we require $C_i^\mathrm{rc}(\uu,t)$,  which is given in the following theorem.

	\begin{theorem} \label{the:1}
		The concentration  of the type $i\in \left\{\A,\B,\C \right\}$ molecules at time $t$ and position $\uu$ assuming that   only reactions occur and diffusion is absent is given by \eqref{eq:5} at the bottom of this page,
		where 
		\setcounter{equation}{9}
		\begin{align}
		&c_{11}(\uu)=C_\A(\uu,t)-C_\B(\uu,t), \\
		&c_{12}(\uu)=C_\A(\uu,t)+C_\C(\uu,t), \\
		&c_2(\uu)=\sqrt{\left(-\kf c_{11}(\uu)+\kb\right)^2+4 \kf \kb c_{12}(\uu)},\\ &c_3(\uu)=C_\A(\uu,t)+C_\B(\uu,t), \\
	&c_4(\uu)=\frac{c_2(\uu)-\kf c_3(\uu)-\kb}{c_2(\uu)+\kf c_3(\uu) +\kb}.
	\end{align}
	\end{theorem}

	\begin{proof}
	Please refer to Appendix~\ref{app:2}.
	\end{proof}
	
	 In some applications,  the backward reaction can be very slow with respect to the time scale of interest. When   $\kb\rightarrow 0$ and $C_\A(\uu,t)=C_\B(\uu,t)$  hold, \eqref{eq:5a} and \eqref{eq:5b} have indeterminate forms. For this case, the values of $C_i^\mathrm{rc}(\uu,t+\Delta t)$  are given in the following corollary.
	\begin{corollary} \label{cor:2}
			For  $\kb=0$ and $C_\A(\uu,t=0)=C_\B(\uu,t=0)$, we have
		\begin{align} \label{eq:8a}
		C_\A^\mathrm{rc}(\uu,t+\Delta t)=&C_\B^\mathrm{rc}(\uu,t+\Delta t)=\frac{ C_\A(\uu,t)}{1+\kf \Delta t C_\A(\uu,t)}
		\end{align}
		and $C_\C^\mathrm{rc}(\uu,t+\Delta t)$ is still given by \eqref{eq:5c}.
	\end{corollary}
	
	\begin{proof}
		When $\kb\rightarrow 0$ and $C_\A(\uu,t)=C_\B(\uu,t)$, \eqref{eq:8a} is obtained by using L'Hospital's rule in \eqref{eq:5a} and \eqref{eq:5b} for $c_{11}(\uu)\rightarrow 0$.
	\end{proof}

\setcounter{eqnback1}{\value{equation}} \setcounter{equation}{19}
\begin{figure*}[!b]
	
	\hrulefill
	\begin{IEEEeqnarray}{lll} \label{eq:56a}
		W_{\B}^z(z,\tilde{z})=\frac{1}{2}+\sum_{n=1}^\infty \cos\left(\frac{n\pi z}{b}\right)\cos\left(\frac{n\pi \tilde{z}}{b}\right)\exp\left(-\frac{n^2\pi^2}{b^2}D \Delta t\right),\IEEEyesnumber\IEEEyessubnumber \label{eq:56}\\
		W_{\B}^\rho(\rho,\tilde{\rho})=\tilde{\rho}\left(1+\sum_{j=1}^\infty \frac{1 }{J_0^2\left(l_{j} \right)}J_0\left(\frac{l_{j}}{a}\rho \right)J_0\left(\frac{l_{j}}{a}\tilde{\rho} \right)\exp\left(-\frac{l_{j}^2}{a^2}D \Delta t\right)\right), \IEEEyessubnumber \label{eq:57}
	\end{IEEEeqnarray} 

\end{figure*}
	
 \subsection{Detection with Unrestricted and Steady-Concentration Probes} \label{sub3:2}
 
 In this subsection, we consider the special case when $C_\B(\uu,t)$ is  very large and thus is not changed significantly by the reaction over time, i.e., $C_\B(\uu,t)$ is assumed to be constant over time. This assumption is similar to the assumption of constant enzyme concentration made in \cite{NCS:14:INB}. The assumption is applicable when the type $\B$ molecules have been released continuously over time from a position $\uu_\B$  such that a steady state is reached at the beginning of  information transmission. The steady state value of $C_\B(\uu,t)$   is given  by
 \setcounter{equation}{15}
 \begin{align} \label{eq:33}
 C_\B(\uu)&=\underset{t\rightarrow \infty}{\lim} \int_0^t C_\B(\uu,\tilde{t})~\diff \tilde{t}\\\nonumber
 &=\underset{t\rightarrow \infty}{\lim}\frac{N_\B}{4\pi D_\B ||\uu-\uu_\B||}\erfc\left(\frac{||\uu-\uu_\B||}{\sqrt{4 D_\B t}}\right)\\       \nonumber                          
 &= \frac{N_\B}{4\pi D_\B ||\uu-\uu_\B||},\nonumber
 \end{align}   
 where $\erfc(\cdot)$ is  the complementary error function.
 Then, if $\kb=0$, and $D_\A=D_\C$, we can obtain  closed-form expressions for $C_\A(\uu,t)$ and $C_\C(\uu,t)$, as given in the following corollary.
 
 \begin{corollary} \label{col:4}
 	Under the above assumption, the concentrations of the type $\A$ and $\C$ molecules are given respectively by
 	\begin{align}
 	\label{eq:34a}
 	C_\A(\uu,t)&=	\frac{N_\A}{(4\pi D_\A t)^{3/2}}\exp\left(-\frac{||\uu-\uu_\A||^2}{4 D_\A t}-\kf C_\B(\uu) t\right),\\
 	\label{eq:34b}
 	C_\C(\uu,t)&= \frac{N_\A}{(4\pi D_\A t)^{3/2}}\exp\left(-\frac{||\uu-\uu_\A||^2}{4 D_\A t}\right)-C_\A(\uu,t),
 	\end{align}
 	where $\uu_\A$ is the release position of the type $\A$ molecules.
 \end{corollary}
 \begin{proof}
 	Please refer to Appendix~\ref{app:3}.                                                                    	\end{proof}

 \begin{remark}
 	Due to \eqref{eq:27c}, $C_\C(\uu,t)$ in \eqref{eq:34b} increases rapidly when $C_\B(\uu,t)$ is large.   For the general case, where $C_\B(\uu,t)$ reduces over time, $C_\C(\uu,t)$ given  in \eqref{eq:34b} is an upper bound.
 \end{remark}

	\subsection{Detection with Restricted Probes }
	
		To use the probe molecules efficiently, we consider an MC system where the molecular probes, i.e., type $\B$ molecules, are restricted to  volume $\Omega$ around the receiver. Hence, the type $\B$ molecules will not disperse in the environment and  more type $\C$ molecules are produced around the receiver. 
		Note that, type $\A$ and $\C$ molecules can diffuse in an unbounded environment. 
		In this case, the concentration of type $\A, \B$, and $\C$  molecules  can still be calculated with Algorithm~\ref{al:1} except that $C_{\B}(\uu,t)=0, \uu \notin \Omega$, and  $	C_{\B}^\mathrm{df}(\uu,t), \uu \in \Omega$, is updated as specified in the following Corollary.
		
		\begin{corollary}\label{cor:5}
			Using cylindrical coordinates, for a bounded volume, $C_{\B}^\mathrm{df}(\uu,t+\dt), \uu \in \Omega,$ is  given by
			\begin{align}\label{eq:55}
			C_{\B}^\mathrm{df}(\uu,t+\dt)&=\frac{4}{ a^2b}\int_{\tilde{\rho}=0}^{a}\int_{\tilde{z}=-\frac{b}{2}}^{\frac{b}{2}} C_{\B}(\tilde{\rho},\tilde{z},t)W_{\B}^z(z,\tilde{z})\nonumber\\
			&\quad \times W_{\B}^\rho(\rho,\tilde{\rho}) ~\diff \tilde{z}\diff \tilde{\rho},
			\end{align}
			where $W_{\B}^z(z,\tilde{z})$ and $W_{\B}^\rho(\rho,\tilde{\rho})$ are given by \eqref{eq:56a} at the bottom of this page,
			 $J_0(\cdot)$ is the zeroth order  Bessel function of the first kind, and $l_{j}$ is the positive root of $J'_0(l_{j})=0$.
		\end{corollary}	
		
\begin{proof}
	Please refer to Appendix~\ref{app:4}.                                                            
	        	\end{proof}
	
	Note that, for a given amount of released type $\B$ molecules, the smaller volume $\Omega$ is, the higher the concentration of the type $\B$ molecules in the volume $\Omega$ is. For a small volume $\Omega$, a slight increase in the released number  of type $\B$ molecules can result in a large increase in the concentration of the type $\B$ molecules  in the volume. This leads to a large increase in the concentration of the created type $\C$ molecules.

\subsection{Threshold Decision Rule} 
Let $s_n$ and $\hat{s}_n$ ($s_n,\hat{s}_n\in\left\{0,1\right\}$) denote the $n$-th transmitted bit and the $n$-th detected bit, respectively.  We adopt a simple threshold decision rule at the receiver where the receiver makes the decision on the transmitted bit based on a signal which is proportional to the number  of  type $\C$ molecules in its volume, denoted by $q$,  at the sampling time, denoted by $t_s$, as follows
\setcounter{equation}{20}
\begin{align} \label{eq:32}
\hat{s}_n=
\begin{cases}
0 \text{ if } q \leq \gamma, \\
1 \text{ if } q> \gamma,
\end{cases}
\end{align}
where $\gamma$ is the decision threshold.
We assume that the movements of the molecules are mutually independent,  and thus, $q$ approximately follows a Poisson distribution \cite{JAW:18:Arxiv} as follows
\begin{align} \label{eq:50}
q \sim \mathcal{P} \left(\bar{q}\right),
\end{align}
 where $\bar{q}$ is the mean of $q$ and given by \cite{JAS:17:CL} 
\begin{align}\label{eq:29}
\bar{q}=\int_{\uu\in \mathcal{V}^\Rx} C_\C(\uu,t_s)\diff \uu.
\end{align}
Note that since \eqref{eq:27} includes the impact of all releases, $C_\C(\uu,t_s)$ and thus $\bar{q}$ is affected by ISI and have different values for different sequences of $n$ bits, denoted by $\mathbf{s}_n=[s_1,s_2,\dots,s_n]$.  

The BER is  given by 
\begin{align}
\label{eq:30}
\Pe=\frac{1}{2}\Big(\Pr\left( q\leq \gamma| s_n=1\right)+1-\Pr\left( q\leq \gamma| s_n=0\right)\Big),
\end{align}
where the cumulative distribution function of the Poisson distribution can be expressed as
\begin{align}
\label{eq:31}
\Pr\left( q\leq \gamma| s_n\right)=\frac{1}{2^{n-1}}\sum_{\mathbf{s}_{n-1}\in\mathbb{S}}\bigg(\exp\left(-\bar{q}\right)\sum_{w=0}^\gamma\frac{\bar{q}^w}{w!}\bigg).
\end{align}
Here, $\mathbb{S}$ is the set of all possible values of $\mathbf{s}_{n-1}$ which affect $\bar{q}$.

	\section{Simulation Results}\label{sec:4}

	\begin{table*}[t!]
		\captionof{table}{System parameters used for the numerical results.}
		\centering
		\label{tab:1}
		\begin{tabular}{c|c||c|c}
			\toprule
			Parameter	&	Value	&	Parameter	&	Value \\[-0.05cm]
			\midrule
			$D_\A$ [$\si{\meter^2\per\second}$]	                                       &	$10\times 10^{-10}$		& 
			$\kf$	[\si[inter-unit-product = \ensuremath{{}\cdot{}}]
			{molecules^{-1}.\meter^3.\second^{-1}}]  &	$10^{-22}$	 	        \\
			$D_\B$ [$\si{\meter^2\per\second}$] 	   &    $ 1.1\times 10^{-10}$	            &
				$\kb$ [$\si{\second^{-1}}$]  &	$10^{-26}$ 	            \\
			$D_\C$ [$\si{\meter^2\per\second}$] 	                                       &    $ 10^{-10}$	            &
			$b=2a$ [$\si{\meter}$]  &	$ 10^{-5} $ 	            \\
			$\dt$ [$\si{\second}$]	                                           &	$10^{-2}$               &
			$T$ [$\si{\second}$]	                                           &	$10$		            \\
			$\mathcal{V}_\Rx$ [$\si{\meter^3}$]	                                               &	$9.8\times 10^{-20}$       & 
			$R$ [$\si{\meter}$]	                                               &	$5\times 10^{-5}$       \\
			$N_\A [\si{molecules}]$                                        &    $5\times 10^8$          &
			$z_{\max}$ [$\si{\meter}$]	                                               &	$6R=3\times 10^{-4}$       \\ 
			$\uu_\Tx$ [$\si{\meter}$]	                                               &	$(0,0,R)$       & 
			$\uu_\Rx$ [$\si{\meter}$]                                        &    $(0,0,0)$          \\
			\bottomrule
		\end{tabular}
	\end{table*}

	In this section, we first confirm the accuracy of Algorithm~\ref{al:1} by  particle-based simulation. We   then use  Algorithm~\ref{al:1}  to analyze the concentration of the type $\C$ molecules for different scenarios. We also evaluate the system performance in terms of the BER and use Monte-Carlo simulation to confirm the analytical results.

	We simulate the system in a bounded environment using cylindrical coordinates with the limits for $\rho$ and $z$ large enough  to approximate an unbounded environment. Let  $z_{\max}$ characterize the boundary of the environment such that $0\leq \rho \leq z_{\max}, -z_{\max}\leq z \leq z_{\max}$.
	For all numerical results presented, we use the 
	parameters provided in Table I, unless otherwise stated. 
	%
	For fast detection, we consider fast forward reactions, e.g., reactions with half-life time on the order of minutes or seconds. The half-life time of a reaction, denoted by $t_{1/2}$, is defined as the time  for the concentration of the reactant to decrease to half of its original  value \cite{Cha:05:Book}, assuming the reactant is uniformly distributed. To select  suitable  parameter orders, we consider the case when the reactants are uniformly distributed and  assume that  the most significant change of concentration results from the forward reaction in \eqref{eq:1}, $t_{1/2}=\SI{1}{\second}$, and $C_\A(\uu,t) \ll C_\B(\uu,t)$, so that the type $\A$ molecules can react and be converted to type $\C$ molecules without a noticeable reduction of the number of type $\B$ molecules. Then, from \cite[Eq.~(9.14)]{Cha:05:Book}, we have 
	\begin{align} \label{eq:40}
	\kf t_{1/2}=\frac{1}{C_\B^0}
	\ln\frac{\left(C_\B^0-\frac{1}{2}C_\A^0\right) C_\A^0}{\frac{1}{2}C_\A^0 C_\B^0}\simeq\frac{\ln 2}{C_\B^0},
	\end{align}
	where $C_i^0=C_i(\uu,t=0)$.
	Adopting $C_\B^0= \SI{6e21}{molecules\per\meter^3}$ and binding constant $K_a=\kf/\kb=\SI[inter-unit-product = \ensuremath{{}\cdot{}}]{6e4}{molecules^{-1}.\meter^3}$ from \cite{LYY:18:CO},   based on \eqref{eq:40}, we chose $\kf$ and $\kb$ as in Table~\ref{tab:1} such that the resulting $K_a$ is on the same order as  $K_a$ in \cite{LYY:18:CO}. 
	
		\begin{figure}[t!] 
			\centering
			\includegraphics[width=\fscale\textwidth]{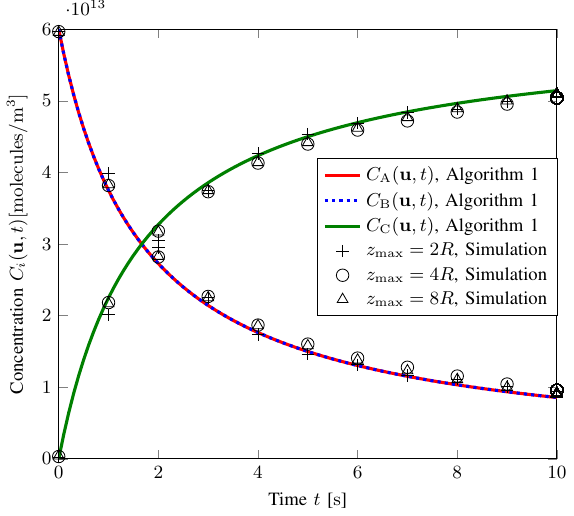}
			\caption{
				Concentrations of the type $\A, \B$,  and $\C$ molecules versus time, where the type $\A$ and $\B$ molecules are uniformly distributed in an approximately-unbounded environment limited by $z_{\max}$.}
			\label{fig:1}
		\end{figure}

\subsection{Verification of the Proposed Algorithm and Poisson Model}

		In Fig.~\ref{fig:1}, we use the particle-based simulation described in   \cite[Appendix~F]{JFS:19:MBSC} to confirm the accuracy of Algorithm~\ref{al:1}, i.e., the solution of the reaction diffusion equation \eqref{eq:27}. We assume that the type $\A$,  $\B$, and $\C$ molecules are uniformly distributed with $C_\A(\uu,t=0)=C_\B(\uu,t=0)=\SI{6e13}{molecules\per\meter^3}$, and $C_\C(\uu,t=0)=\SI{0}{molecules\per\meter^3}$, and they have the same diffusion coefficient $D_\A$ given in Table~1  such that  diffusion does not have any impact on the concentrations of the type $\A, \B$, and $\C$ molecules \footnote{If $D_\A>D_\B$, for example, after a reaction occured at a location, type $\A$ molecules may diffuse to that location faster than type $\B$ molecules and the concentrations of type $\A$ and $\B$ molecules at that location can be different.}. In order to reduce the computational complexity for particle-based simulation, we choose $\kf=\SI[inter-unit-product = \ensuremath{{}\cdot{}}]{e-14}{molecules^{-1}.\meter^3.\second^{-1}}$ and $\kb=\SI{e-18}{\second^{-1}}$. In Fig.~\ref{fig:1}, since the value of $\kf$ is larger than that of $\kb$ and $C_\A(\uu,t=0)=C_\B(\uu,t=0)$,  we observe that $C_\A(\uu,t)$ and $C_\B(\uu,t)$ are equal  and decrease over time while $C_\C(\uu,t)$ increases over time as expected. In general, the results obtained with Algorithm~\ref{al:1} are in good agreement with the simulation results. The simulation results become more accurate for larger $z_{\max}$, when the assumption of an unbounded environment becomes more justified.

In Fig.~\ref{fig:5}, we use particle-based simulation to confirm the accuracy of Corollary~\ref{col:4}, i.e., the solution of the diffusion equation in a bounded environment. We simulate the diffusion of the type $\B$ molecules with $N_\B=10^8 \si{molecules}$ and adopt  $b=2a= 10^{-5}\si{\meter} $. We compare the concentrations of the type $\B$ molecules  at the release point $\uu=\uu_\B=(0,0,0)$, half way to the boundary $\uu=(0,0,b/4)$, and at the boundary $\uu=(0,0,b/2)$. Due to the bounded environment, the concentration of the type $\B$ molecules  reaches a steady state, i.e., a constant value  everywhere in the environment, after a period of time. As expected, at the release point $\uu=(0,0,0)$, $C_\B$ decreases from the highest value at the  time of release to the steady value. At the half way point to the boundary $\uu=(0,0,b/4)$, $C_\B$ first increases from zero to a maximum, and then decreases to the steady value. At the boundary $\uu=(0,0,b/2)$, $C_\B$ increases from zero to the steady value. We also show $C_\B$ at the release point for an unbounded environments. The two curves for $C_\B$ at the release point for the bounded and unbounded environments are identical for small $t$. However, $C_\B$ for the unbounded environment  decreases to zero for large $t$. In general, the results obtained by Corollary~\ref{col:4} are in excellent agreement with the simulation results. 

\begin{figure}[!t]
	\centering
	\includegraphics[width=\fscale\textwidth]{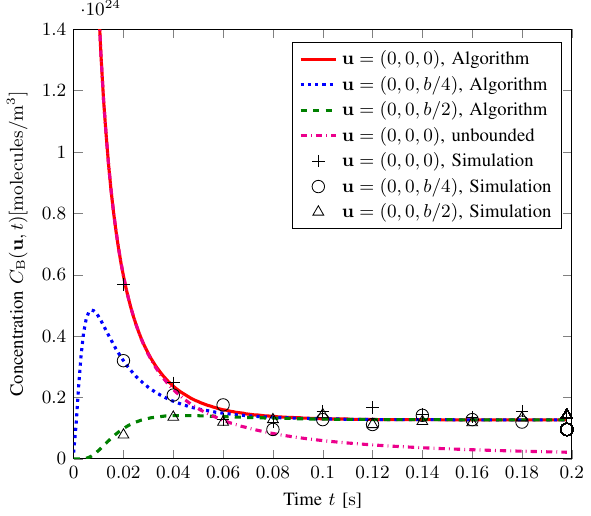}
	\caption{
		Concentrations of type $\B$ molecules at different positions in unbounded and bounded volumes.}
	\label{fig:5}
\end{figure}

In Fig.~\ref{fig:Poi}, we use particle-based simulation to validate the assumption that the number $q$ of type $\C$ molecules in the receiver volume is Poisson distributed, see \eqref{eq:50}. We use the same simulation setup as  for Fig.~\ref{fig:1} to obtain the histogram for $q$, i.e., an estimate for the true distribution. We compare this result with the probability mass function (PMF) of the Poisson distribution with the mean given in \eqref{eq:29}, $C_{\C}(\uu,t)$ obtained by using Algorithm~1, and a receiver volume  $\mathcal{V}^\Rx$ equal to $\SI{5.24e-13}{\meter^3}$.  In particular, in Fig.~\ref{fig:Poi}, we show the PMF of $q$ observed at times $t_s=[1,2,3] \si{\second}$. The corresponding mean values used for the PMF are $11.8, 17.1$, and $ 20.2$, respectively, which are obtained from $C_{\C}(\uu,t)$ shown in Fig.~\ref{fig:1}. We observe that the Poisson distributions with the mean obtained with Algorithm~1 are in good agreement with the histograms obtained by  particle-based simulation. Thus, the assumption of the Poisson model for the type $\C$ product molecules is justified.

		\begin{figure}[!t]
			\centering
			\includegraphics[width=\fscale\textwidth]{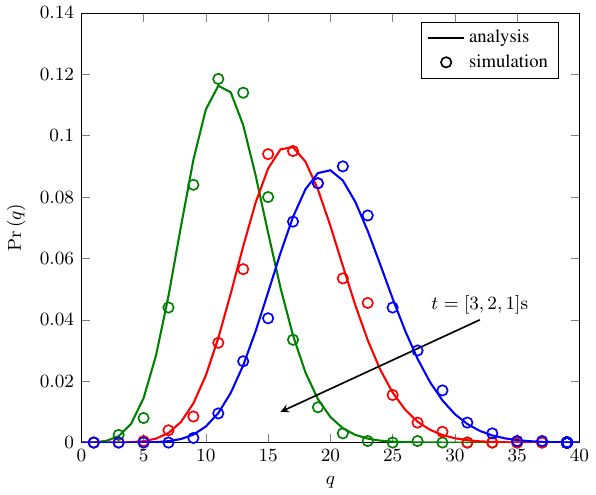}
			\caption{
				PMFs of the number $q$ of type $\C$ molecules in receiver volume $\mathcal{V}^\Rx=\SI{5.24e-13}{\meter}$ at three sample times $t_s=[1,2,3] \si{\second}$.}
			\label{fig:Poi}
		\end{figure}

			\begin{figure}
				\centering
				\includegraphics[width=\fscale\textwidth]{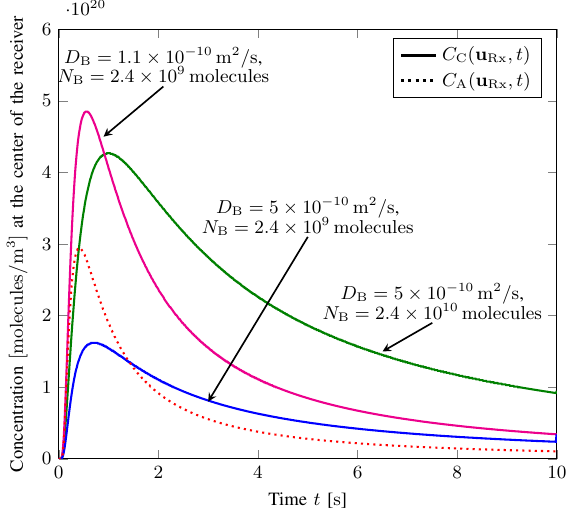}
				\caption{Concentration of type $\A$ and $\C$ molecules at the center of the receiver versus time for direct  and indirect detection, respectively, when the type $\B$ molecules are not restricted. For the latter case, different values of $D_\B$ and $N_\B$  are considered.}
				\label{fig:2}
			\end{figure}

		\subsection{Impact of System Parameters on  Concentration of Type $\C$ Molecules} 
	
		In Fig.~\ref{fig:2}, we present  the concentrations of the type $\A$ and $\C$ molecules at the center of the receiver versus time for   direct and indirect detection, respectively, when the type $\B$ molecules are not restricted. For the proposed indirect detection, we study the impact of the diffusion coefficient and the number of type $\B$ molecules released at the center of the receiver on the concentration of the type $\C$ molecules. In particular, we assume $\uu_\B=\uu_\Rx$, $D_\B=[1.1, 5] \times 10^{-10}\si{\meter^2\per \second}$, and $N_\B=[2.4, 24]\times 10^9\si{molecules} $.  We observe that when the type $\A$ molecules cannot be detected directly, and thus, the proposed indirect detection is used, the  concentration of the type $\C$ molecules  at the center of the receiver, $C_\C(\uu_\Rx,t)$, has a similar characteristic, i.e.,  a single peak and a long tail, as $C_\A(\uu_\Rx,t)$ when the type $\A$ molecules can be directly detected. Let $\max_t C_i(\uu_\Rx,t), i\in\left\{\A,\C\right\}$, denote the peak value of $C_i(\uu_\Rx,t)$, i.e., the maximum value of $C_i(\uu_\Rx,t)$ over time.
		For larger $N_\B$ and a given $D_\B$, e.g., $D_\B=\SI{5e-10}{\meter^2\per \second}$, $\max_t C_\C(\uu_\Rx,t)$ is larger. This is expected since a larger number of type $\B$ molecules  produce a larger number of type $\C$ molecules at the receiver. $\max_t C_\C(\uu_\Rx,t)$ can even exceed $\max_t C_\A(\uu_\Rx,t)$ because $D_\C<D_\A$ and type $\C$ molecules diffuse away from the receiver more slowly than type $\A$ molecules. However, the tail of $C_\C(\uu_\Rx,t)$ can be  heavier than that of $C_\A(\uu_\Rx,t)$. In particular, for $N_\B=\SI{2.4e9}{molecules} $, although $\max_t C_\C(\uu_\Rx,t)<\max_t C_\A(\uu_\Rx,t)$, the tail of $C_\C(\uu_\Rx,t)$ is heavier than that of $C_\A(\uu_\Rx,t)$ for the considered range of time $t$. This may negatively effect  system performance due to the increased level of ISI as will be shown in Fig.~\ref{fig:4}.
		  Nevertheless, using the type $\C$ molecules for detection is unavoidable when direct detection of the type $\A$ molecules is \emph{impossible}. Moreover, for a given $N_\B$, e.g., $N_\B=\SI{2.4e9}{molecules} $,  and smaller $D_\B$, e.g., $D_\B=\SI{1.1e-10}{\meter^2\per \second}$, $\max_t C_\C(\uu_\Rx,t)$ is also larger. Smaller $D_\B$ result in slower diffusion of type $\B$ molecules from the receiver. Thus, there are more type $\B$ molecules near the receiver to react  and produce more type $\C$ molecules. However, due to the larger number of reactions, the amount of the type $\B$ molecules in the environment reduces significantly and thus less type $\C$ molecules are produced at later times, which results in a lighter tail of $C_\C(\uu_\Rx,t)$, i.e., less ISI. 
		  
		  	 	\begin{figure}[!t]
		  	 		\centering
		  	 		\includegraphics[width=\fscale\textwidth]{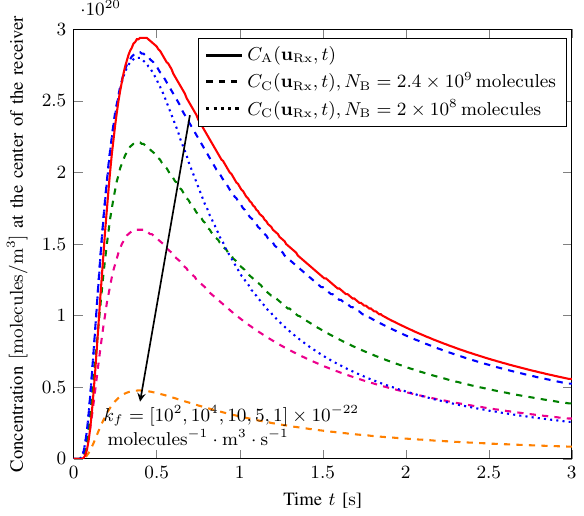}
		  	 		\caption{
		  	 			Concentration of  type $\A$ and $\C$ molecules at the center of the receiver versus time for  direct and indirect detection, respectively, when the type $\B$ molecules are not restricted. }
		  	 		\label{fig:kf1}
		  	 	\end{figure}
	In Fig.~\ref{fig:kf1}, we investigate the impact of $k_f$ and $N_\B$ on $C_{\C}(\uu_\Rx,t)$ when the type $\B$ molecules are not restricted. In particular, we present the concentrations of  the type $\A$ and $\C$ molecules at the center of the receiver versus time for direct detection and indirect detection with $k_b=0$, $D_\A=D_\B=D_\C=\SI{1e-9}{\meter^2\per \second}$. We consider	$N_{\B}=\SI{2.4e9}{molecules} $ and $k_f=[100,10,5,1]\times10^{-22}\si[inter-unit-product = \ensuremath{{}\cdot{}}]
	{molecules^{-1}.\meter^3.\second^{-1}}$ as well as	$N_{\B}=\SI{2e8}{molecules} $ and $k_f=10^{-18}\si[inter-unit-product = \ensuremath{{}\cdot{}}]
	{molecules^{-1}.\meter^3.\second^{-1}}$. We observe that the concentration of the type $\C$ molecules increases quickly after the release of the type $\A$ and $\B$ molecules and the peak value of the concentration increases when $k_f$ increases. This means more type $\C$ molecules are created when the reaction rate is larger. The concentration of the type $\C$ molecules approaches the concentration of the type $\A$ molecules for large $k_f$. For $N_{\B}=\SI{2.4e9}{molecules} $ and $k_f=10^{-20}\si[inter-unit-product = \ensuremath{{}\cdot{}}]	{molecules^{-1}.\meter^3.\second^{-1}}$, $C_{\C}(\uu_\Rx,t)\approx C_{\A}(\uu_\Rx,t)$.  This means $k_f$ is large enough such that the reaction between type $\A$ and $\B$ molecules happens immediately whenever they come close to each other. Thus, most of the type $\A$ molecules will become type $\C$ molecules. Increasing $k_f$ further would not result in any significant change of $C_{\C}(\uu_\Rx,t)$. Note that the immediate reaction alone is not enough to yield  $C_{\C}(\uu_\Rx,t)\approx C_{\A}(\uu_\Rx,t)$ as $D_\A=D_\C$ is also needed for the created type $\C$ molecules to diffuse in the same manner  as the type $\A$ molecules.  For a smaller $N_{\B}$, e.g., $N_{\B}=\SI{2e8}{molecules} $, and sufficiently large $k_f$, e.g., $k_f=10^{-18}\si[inter-unit-product = \ensuremath{{}\cdot{}}]
	{molecules^{-1}.\meter^3.\second^{-1}}$, the peak values of $C_{\C}(\uu_\Rx,t)$ and $ C_{\A}(\uu_\Rx,t)$ are approximately identical but the tail of  $C_{\C}(\uu_\Rx,t)$ decays faster than the tail of $ C_{\A}(\uu_\Rx,t)$. Here, not all type $\A$ molecules are converted to type $\C$ molecules because there are not enough type $\B$ molecules available to react with type $\A$ molecules at  later times.

		  	 	\begin{figure}[!t]
		  	 		\centering
		  	 		\includegraphics[width=\fscale\textwidth]{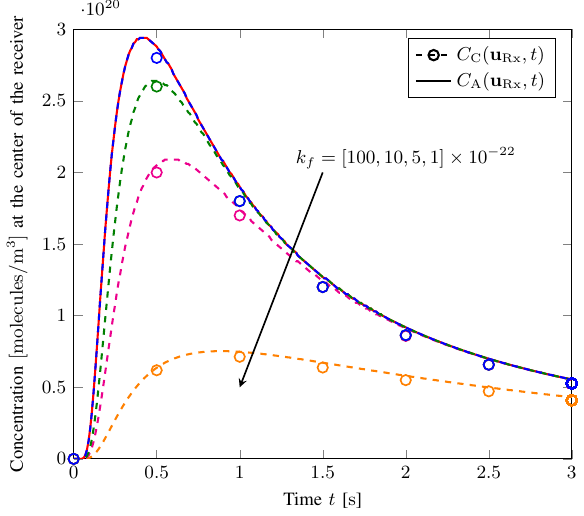}
		  	 		\caption{
		  	 			Concentration of  the type $\A$ and $\C$ molecules at the center of the receiver versus time for direct and indirect detection, respectively, when the concentration of the unrestricted type $\B$ molecules is constant. For $C_{\C}(\uu_\Rx,t)$, the dashed lines and markers denote the closed-form expression in \eqref{eq:34b} and analysis obtained with Algorithm~1, respectively.  }
		  	 		\label{fig:Buni1}
		  	 	\end{figure}
		  	 	
In Fig.~\ref{fig:Buni1}, we study the impact of $k_f$  on $C_{\C}(\uu_\Rx,t)$ for the special setup considered in Subsection~\ref{sub3:2}, where  the  concentration of the type $\B$ molecules is assumed to be large and thus remains unchanged when  reacting with the type $\A$ molecules. We further assume $k_b=0$, $D_\A=D_\C=\SI{1e-9}{\meter^2\per \second}$, $ C_{\B}(\uu_\Rx)=\SI{5e-21}{molecules\per\meter^3}$, and $k_f$ varies. The dashed lines represent the closed-form expression of $C_{\C}(\uu_\Rx,t)$ in \eqref{eq:34b}. The markers denote $C_{\C}(\uu_\Rx,t)$ obtained with Algorithm~1. The solid lines represent  $C_{\A}(\uu_\Rx,t)$ for direct detection. First, we observe  that the derived close-form results are in excellent agreement with the results obtained by Algorithm~1. We also observe  that $C_{\C}(\uu_\Rx,t)$  approaches $C_{\A}(\uu_\Rx,t)$ for large $k_f$, as in Fig.~\ref{fig:kf1}.  Interestingly,  $C_{\C}(\uu_\Rx,t)$ can be equal to $C_{\A}(\uu_\Rx,t)$ in Fig.~\ref{fig:Buni1} whereas there is a small difference between them in Fig.~\ref{fig:kf1}.  This is because for the case considered in Fig.~\ref{fig:Buni1}, $C_{\B}(\uu_\Rx)$ remains constant and all type $\A$ molecules can react with type $\B$ molecules to create type $\C$ molecules. 
		 
		 	\begin{figure}[!t]
		 		\centering
		 		\includegraphics[width=\fscale\textwidth]{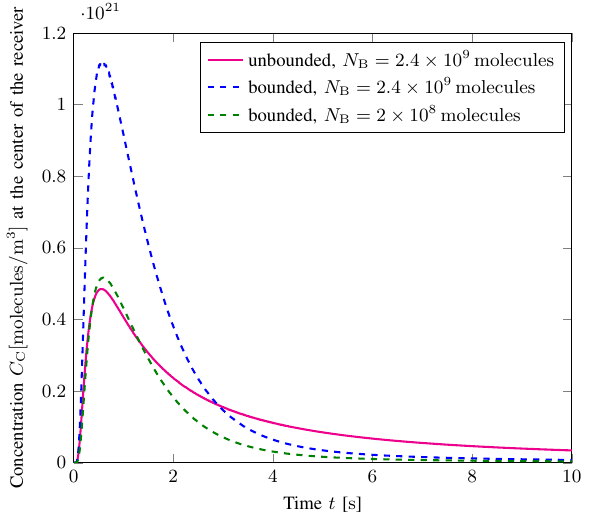}
		 		\caption{
		 			Concentration of  type $\C$ molecules at the center of the receiver versus time when type $\B$ molecules are released in  unbounded and bounded volumes. }
		 		\label{fig:6}
		 	\end{figure}

In Fig.~\ref{fig:6}, we present the concentrations of the type $\C$ molecules at the center of the receiver versus time when the type $\B$ molecules are and are not confined to a bounded volume. For the same number of type $\B$ molecules released, i.e., $N_\B=\SI{2.4e9}{molecules}$, the concentration of the type $\C$ molecules has a much higher peak value and decreases faster when the type  $\B$ molecules are restricted to a bounded volume compared to the case when they are not. For a much smaller number of type $\B$ molecules, i.e., $ N_\B=\SI{2e8}{molecules}$, confined to the bounded volume, the same peak value and a faster decrease of $C_\C$ can be obtained  compared to the case when the type $\B$ molecules  are in an unbounded volume.   Therefore, ISI is less severe when the type $\B$ molecules are restricted in a bounded volume. The reason is that,  outside the bounded volume, no reaction happens and thus no type $\C$ molecules are created  which may arrive at the receiver later and contribute to the ISI.


		 \subsection{ System Performance}
		
				\begin{figure}[!t]
					\centering
					\includegraphics[width=\fscale\textwidth]{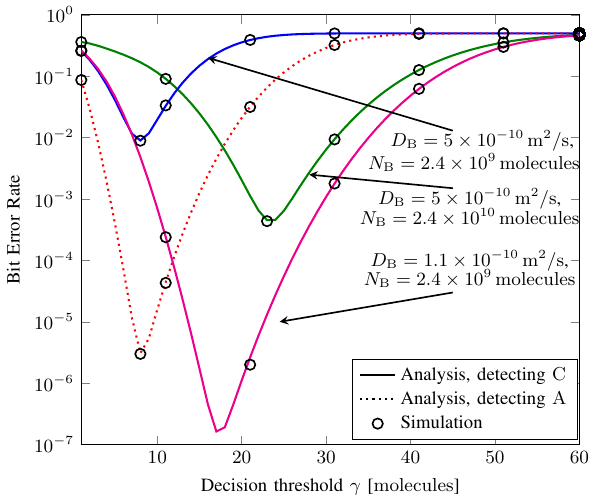}
					\caption{
						BER versus  decision threshold $\gamma$ for direct detection via the type $\A$ molecules and indirect detection via the type $\C$ molecules.}
					\label{fig:4}
				\end{figure}

		Fig.~\ref{fig:4} depicts the BER of the considered MC system versus  decision threshold, $\gamma$, for direct  and indirect detection. The type $\B$ molecules are  released at the center of the receiver.
		We take the ISI caused by the previous two symbols into account, i.e., $s_n$ is interfered by $s_{n-1}$ and $s_{n-2}$, and the bit interval $T=\SI{10}{\second}$ is long enough, such that the contribution of the other previous symbols, e.g., $s_{n-3}$, to the ISI is negligible. We choose the sampling time $t_s$ equal to the time when $C_\C(\uu_\Rx,t)$ assumes its maximum value, $\max_t C_\C(\uu_\Rx,t)$. From Fig.~\ref{fig:4}, we  observe that the analytical results obtained with \eqref{eq:30} and \eqref{eq:31} are in excellent agreement with the corresponding Monte-Carlo simulation results. Furthermore, the BER can be reduced significantly by optimizing the  decision threshold. We also observe that although $ C_\C(\uu_\Rx,t)$ for the case of $D_\B=\SI{1.1e-10}{\meter^2\per \second}$ and $N_\B=2.4\times 10^9$ molecules has the highest peak value in Fig.~\ref{fig:2}, the corresponding optimal decision  threshold is smaller than that for the case of $D_\B=\SI{5e-10}{\meter^2\per \second}$ and $N_\B=2.4\times 10^{10}$ molecules. This is due to the fact that the optimal threshold depends on both the peak value and the tail of $C_\C(\uu_\Rx,t)$. Moreover, when $N_\B$ increases or $D_\B$ decreases, the minimum BER value decreases due to the reduced ISI. When direct detection is not possible, the addition of the type $\B$ molecules makes  detection via the type $\C$ molecules possible even if the resulting BER may be higher compared to the case when direct detection is possible. However, when the released type $\B$ molecules are appropriately chosen, e.g.,  $D_\B=\SI{1.1e-10}{\meter^2\per \second}$ and $N_\B=2.4\times10^9$ molecules, the proposed indirect detection approach can even achieve a lower BER than direct detection. 

	 Fig.~\ref{fig:7} shows the BER obtained with the optimal decision threshold versus the symbol interval length $T$  when the type $\B$ molecules are released in  unbounded and bounded volumes. As $T$ increases, the impact of ISI reduces and thus the BER decreases. A larger number of released type $\B$ molecules results in a lower BER. For the same $N_\B$, the BER with bounded type  $\B$ molecules is lower than that with unbounded type  $\B$ molecules because the ISI is less severe as explained in the discussion of Fig.~\ref{fig:6}. The same BER and  a higher data rate, i.e., smaller time intervals, can be achieved with fewer type $\B$ molecules released in a bounded volume compared to when their movement is not restricted.  This illustrates the efficient use of  resources when the type $\B$ molecules are restricted to a volume around the receiver.

			\begin{figure}[!t]
				\centering
				\includegraphics[width=\fscale\textwidth]{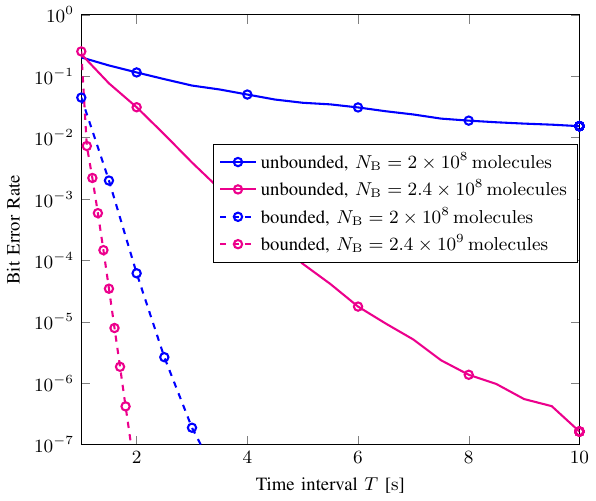}
				\caption{
					BER obtained with the optimal decision threshold versus symbol interval $T$  when type $\B$ molecules are released in  unbounded and bounded volumes.}
				\label{fig:7}
			\end{figure}
			
			Fig.~\ref{fig:8} presents the BER obtained with the optimal decision threshold versus the number of released type $\B$ molecules $N_\B$ for different symbol interval lengths $T$. We observe that for larger $T$, the BER decreases faster as $N_\B$ increases. For the same $T$, for bounded type $\B$ molecules, the BER is lower and decreases faster as $N_\B$ increases compared to  unbounded type $\B$ molecules.
			\begin{figure}[!t]
				\centering
				\includegraphics[width=\fscale\textwidth]{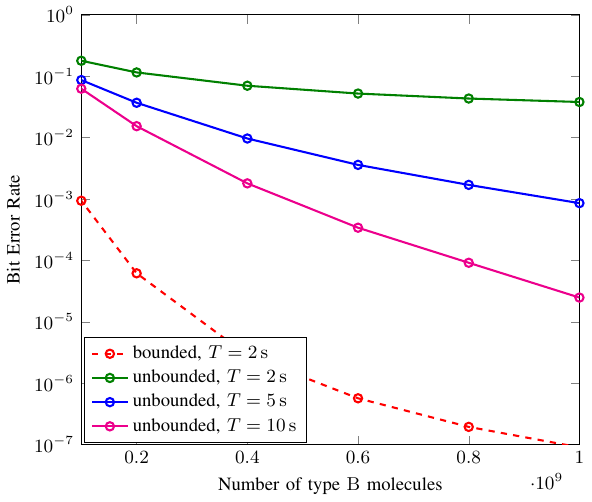}
				\caption{
					BER obtained with the optimal decision threshold versus the number of released type $\B$ molecules $N_\B$ for different symbol interval lengths $T$.}
				\label{fig:8}
			\end{figure}
\section{Conclusions}
In this work, we proposed a novel detection mechanism for  MC systems where the signaling molecules cannot be directly detected at the receiver. Therefore,  a molecular probe was introduced to react with the signaling  molecules to produce  product molecules that can then be detected at the receiver. The molecular probes needed for the proposed detection mechanism can be deployed differently, i.e., as unrestricted probes and as 	probes restricted in a small volume around the receiver. For system performance analysis, we developed an efficient iterative algorithm to find the spatio-temporal concentration of the product molecules taking into account diffusion and reactions. For the special case where the probe concentration is constant over time, closed-form expressions for the concentrations were derived. Our results showed that  the concentration of the product molecules exhibits a similar characteristic over time as the concentration of the signaling molecules.  We  analyzed the  performance of the MC system using the proposed detection scheme in terms of  BER.  Our results  showed that the BER for indirect detection can be significantly improved by optimizing the decision threshold and can even be lower than the BER for direct detection if the molecular probe is suitably chosen. Moreover, when the molecular probe is kept in a volume around the receiver, the same BER and a higher data rate can be achieved with  much fewer molecular probes  compared to when the molecular probes are not restricted.

\appendices
\renewcommand{\thesectiondis}[2]{\Alph{section}:}
\section{Proof of Corollary~\ref{cor:1}}\label{app:1}
We can derive \eqref{eq:35}  by expanding \eqref{eq:25} in cylindrical coordinates with $\tilde{\uu}=(\tilde{\rho},\tilde{\phi},\tilde{z})$ and substituting 
\begin{align}
||\uu-\tilde{\uu}||^2=\left(z-\tilde{z}\right)^2+\rho^2+\tilde{\rho}^2-2\rho\tilde{\rho}\cos\left(\tilde{\phi}\right)
\end{align}
  into \eqref{eq:25}. Then, by using 
  \begin{align}
  \int_{\tilde{\phi}=0}^{2\pi}\exp\left(\frac{2\rho\tilde{\rho}\cos\left(\tilde{\phi}\right)}{4D_i \dt}\right)\diff \tilde{\phi}=2\pi I_0\left(\frac{\rho\tilde{\rho}}{2 D_i \Delta t}\right),
  \end{align}
   we obtain \eqref{eq:35}.

	\section{Proof of Theorem~\ref{the:1}}\label{app:2}
	
	We obtain \eqref{eq:5} by following similar steps as  in \cite[Appendix~D]{JFS:19:MBSC} to solve the following set of equations		
	\begin{IEEEeqnarray}{lll} \label{eq:2}
		\frac{\partial C_\A^\mathrm{rc}(\uu,t)}{\partial t}=-\kf C_\A^\mathrm{rc}(\uu,t) C_\B^\mathrm{rc}(\uu,t) +\kb C_\C^\mathrm{rc}(\uu,t),\IEEEyesnumber\IEEEyessubnumber \label{eq:2a}\\
		\frac{\partial C_\B^\mathrm{rc}(\uu,t)}{\partial t}=-\kf C_\A^\mathrm{rc}(\uu,t) C_\B^\mathrm{rc}(\uu,t) +\kb C_\C^\mathrm{rc}(\uu,t),\IEEEyessubnumber\label{eq:2b}\\
		\frac{\partial C_\C^\mathrm{rc}(\uu,t)}{\partial t}=\kf C_\A^\mathrm{rc}(\uu,t) C_\B^\mathrm{rc}(\uu,t) -\kb C_\C^\mathrm{rc}(\uu,t). \IEEEyessubnumber \label{eq:2c}
	\end{IEEEeqnarray}
	Subtracting \eqref{eq:2b} from \eqref{eq:2a} and adding \eqref{eq:2a} and \eqref{eq:2c}, respectively, we obtain
	\begin{align} \label{eq:42}
	\frac{\partial\left( C_\A^\mathrm{rc}(\uu,t)- C_\B^\mathrm{rc}(\uu,t)\right)}{\partial t}=0,
	\end{align}
		\begin{align}  \label{eq:43}
				\frac{\partial\left( C_\A^\mathrm{rc}(\uu,t)+ C_\C^\mathrm{rc}(\uu,t)\right)}{\partial t}=0.
		\end{align}
	Equations  \eqref{eq:42} and \eqref{eq:43} have  solutions $C_\A^\mathrm{rc}(\uu,t)- C_\B^\mathrm{rc}(\uu,t)=c_{11}(\uu)$ and $C_\A^\mathrm{rc}(\uu,t)+ C_\C^\mathrm{rc}(\uu,t)=c_{12}(\uu)$, where $c_{11}(\uu)=C_\A^\mathrm{rc}(\uu,t=t_0)- C_\B^\mathrm{rc}(\uu,t=t_0)$ and   $c_{12}(\uu)=C_\A^\mathrm{rc}(\uu,t_0)+ C_\C^\mathrm{rc}(\uu,t_0)$, respectively. Here, $t_0$ is the initial time for which the initial conditions are known. Substituting $C_\A^\mathrm{rc}(\uu,t)- C_\B^\mathrm{rc}(\uu,t)=c_{11}(\uu)$ and $C_\A^\mathrm{rc}(\uu,t)+ C_\C^\mathrm{rc}(\uu,t)=c_{12}(\uu)$ into \eqref{eq:2a}, we have
	\begin{align}  \label{eq:44}
	\frac{\partial C_\A^\mathrm{rc}(\uu,t)}{\partial t}&=-\bigg(\kf \left({C_\A^\mathrm{rc}}(\uu,t)\right)^2+\left(-\kf c_{11}(\uu)+\kb\right) \nonumber\\
	&\quad \times C_\A^\mathrm{rc}(\uu,t) -\kb c_{12}(\uu)\bigg),
	\end{align}
	which can be rewritten as
	\begin{align} \label{eq:45}
	&\frac{\partial\left( C_\A^\mathrm{rc}(\uu,t)\right)}{\kf \left({C_\A^\mathrm{rc}}(\uu,t)\right)^2+\left(-\kf c_{11}(\uu)+\kb\right)C_\A^\mathrm{rc}(\uu,t) -\kb c_{12}(\uu)}\nonumber\\
	&\hspace{6cm}=-\partial t.
	\end{align}
	Integrating both sides of \eqref{eq:45} and using \cite{GR:07:Book}
	\begin{align}
	\int \frac{\diff x}{ax^2+bx+c}=\frac{1}{\Delta}\log\left(\frac{\Delta-2ax-b}{\Delta+2ax+b}\right),
	\end{align}
	where $\Delta=\sqrt{b^2-4ac}$, we obtain
	\begin{align} \label{eq:46}
	-t+&\tilde{c}_4(\uu)=\\\nonumber
	&\frac{1}{c_2(\uu)}\ln\left(\frac{c_2(\uu)-2\kf C_\A^\mathrm{rc}(\uu,t)+\kf c_{11}(\uu)-\kb}{c_2(\uu)+2\kf C_\A^\mathrm{rc}(\uu,t)-\kf c_{11}(\uu)+\kb}\right),
	\end{align}
	where $c_2(\uu)=\sqrt{\left(-\kf c_{11}(\uu) +\kb\right)^2+4\kf\kb c_{12}(\uu)}$ and $\tilde{c}_4(\uu)$ is a constant with respect to time. Using the initial condition when $t=t_0$ in \eqref{eq:46} leads to
	\begin{align}
	\label{eq:47}
	\tilde{c}_4(\uu)&=\frac{1}{c_2(\uu)}\ln\left(\frac{c_2(\uu)-\kf c_3(\uu) -\kb}{c_2(\uu)+\kf c_3(\uu)+\kb}\right),
	\end{align} 
	where $c_3(\uu)=C_\A^\mathrm{rc}(\uu,t_0)+C_\B^\mathrm{rc}(\uu,t_0)$.
	Defining $c_4(\uu)=\exp\left(c_2(\uu) \tilde{c}_4(\uu)\right)$,  substituting  \eqref{eq:47} into \eqref{eq:46}, and setting the initial time and the current time, denoted by $t_0$ and $t$ in this proof, equal to $t$ and $\Delta t+t$,  respectively, for each iteration in Algorithm~1, we obtain \eqref{eq:5a}. Using $C_\B^\mathrm{rc}(\uu,t)=C_\A^\mathrm{rc}(\uu,t)- c_{11}(\uu) $ and $ C_\C^\mathrm{rc}(\uu,t)=c_{12}(\uu)-C_\A^\mathrm{rc}(\uu,t)$, it is straightforward to obtain \eqref{eq:5b} and \eqref{eq:5c}, respectively.

	\section{Proof of Corollary~\ref{col:4}}\label{app:3}
	The expression in \eqref{eq:34a} is obtained from \cite[Eq. (9)]{NCS:14:INB} with $k_{-1}=0$. Adding \eqref{eq:27a} and \eqref{eq:27c}, we obtain  
	\begin{align}\label{eq:41}
	\frac{\partial C(\uu,t)}{\partial t}=G_\A(\uu,t)+D_\A \nabla^2 C(\uu,t),
	\end{align}
	where $C(\uu,t)=C_\A(\uu,t)+C_\C(\uu,t)$. From the solution $C(\uu,t)$ of \eqref{eq:41}, we obtain \eqref{eq:34b}.

	\section{Proof of Corollary~\ref{cor:5}}\label{app:4}
	
	As the diffusion and reaction can be decoupled in $\Delta t$, to obtain \eqref{eq:55}, we need to solve the diffusion-only equation 
	\begin{align} \label{eq:58}
		\frac{\partial C(\rho,\phi,z,t)}{\partial t}=D \nabla^2 C(\rho,\phi,z,t),
	\end{align}
	for the following initial and boundary conditions
	\begin{align}\label{eq:59}
	 C(\rho,\phi,z,t=t_0)=f(\rho,\phi, z),
	\end{align}
	\begin{align}	\label{eq:60}
	\frac{\partial C(\rho,\phi,z,t)}{\partial z}\left|_{z=-\frac{b}{2},z=\frac{b}{2}}\right.=0,
	\end{align}
	\begin{align}\label{eq:61}
	\frac{\partial C(\rho,\phi,z,t)}{\partial \rho}\left|_{\rho=a}\right.=0.
	\end{align}
Here, $f(\rho,\phi, z)$ is an arbitrary initial condition at $t_0$. Then, we can obtain $	C_{\B}^\mathrm{df}(\uu,t+\dt)$ as a function of $ 	C_{\B}^\mathrm{df}(\uu,t)$  by solving  \eqref{eq:58} and mapping $ 	C_{\B}^\mathrm{df}(\uu,t)$  and $\Delta t$ to the initial condition and $t$ of  \eqref{eq:58}, respectively.

	Using variable separation, we assume $ C(\rho,\phi,z,t)=P(\rho)\Phi(\phi)Z(z)T(t)$, where $P(\rho)$, $\Phi(\phi)$, $Z(z)$, and $T(t)$ are functions of $\rho$, $\phi$, $z$, and $t$, respectively. From \eqref{eq:58}, we have
	\begin{align} \label{eq:69}
	P(\rho)\Phi(\phi)&Z(z) T'(t)= DT\left(\frac{\Phi(\phi)Z(z)}{\rho}\frac{\partial}{\partial \rho}\left(\rho P'(\rho)\right)\right.\nonumber\\	&\quad\left.+\frac{P(\rho)Z(z)}{\rho^2}\Phi''(\phi)+P(\rho)\Phi(\phi)Z''(z)\right),
	\end{align}
where $F'(x)$ and $F''(x)$ denote first and second derivative of $F(x)$, respectively. 
	Simplifying \eqref{eq:69}, we have 
	\begin{align}
		 T'(t)=	DT&\left(\frac{1}{\rho P(\rho)}\left(P'(\rho)+\rho P''(\rho)\right)\right.\nonumber\\
		 &\quad+ \left.\frac{1}{\Phi(\phi)\rho^2}\Phi''(\phi)+\frac{Z''(z)}{Z(z)}\right)
	\end{align}
	$\Leftrightarrow$
	\begin{align}
		 \frac{1}{\rho P(\rho)}\left(P'(\rho)+\rho P''(\rho)\right)&+\frac{1}{\Phi(\phi)\rho^2}\Phi''(\phi)+\frac{Z''(z)}{Z(z)}\nonumber\\
		 &\quad-\frac{T'(t)}{DT(t)}=0
	\end{align}
		$	\Leftrightarrow$
					\begin{IEEEeqnarray}{ll} \label{eq:67}
							\frac{1}{\rho P(\rho)}\left(P'(\rho)+\rho P''(\rho)\right)+l^2=0,\IEEEyesnumber\IEEEyessubnumber \label{eq:67a}\\
							\frac{1}{\Phi(\phi)\rho^2}\Phi''(\phi)=0,\IEEEyessubnumber\label{eq:67b}\\
							\frac{Z''(z)}{Z(z)}=k^2, \IEEEyessubnumber \label{eq:67c}\\
							-\frac{T'(t)}{DT(t)}=h^2, \IEEEyessubnumber \label{eq:67d}\\
							k^2+h^2=l^2. \IEEEyessubnumber \label{eq:67e}
				\end{IEEEeqnarray}
					
		 Eq.~\eqref{eq:67b} is obtained due to the fact that the system is symmetric and the concentration does not depend on $\phi$, i.e., $ \Phi(\phi)$ is a constant.	The solution of \eqref{eq:67c} for initial condition \eqref{eq:60} is $Z(z)=A \cos(\sqrt{-k^2}z)=A\cos\left(\frac{n\pi}{b}z\right)$, where $-k^2=\frac{n^2\pi^2}{b^2}$, $n=0, 1, 2, 3, \dots$,  and $A$ is a constant.  The solution of \eqref{eq:67a} for initial condition \eqref{eq:61} is 
			 $ P(\rho)=J_0\left(\frac{l_{j}}{a}\rho \right)$, where $l_{j}$ satisfies  $J'_0\left(l_{j} \right)=0$ and  $j=0, 1, 2, 3, \dots$. Then, due to \eqref{eq:67e}, the solution of \eqref{eq:67d} is $T(t)=C e^{-h^2Dt}=C e^{\left(-\frac{n^2\pi^2}{b^2}-\frac{l_{j}^2}{a^2}\right)Dt}$, where $C$ is a constant. Hence, we have
			 \begin{align}\label{eq:70}
			 &C(\rho,\phi,z,t)=\\\nonumber&\hspace{1cm}\sum_{n=0}^\infty\sum_{j=0}^\infty a_{nj} J_0\left(\frac{l_{j}}{a}\rho \right)
			 \cos\left(\frac{n\pi z}{b}\right)
			  e^{\left(-\frac{n^2\pi^2}{b^2}-\frac{l_{j}^2}{a^2}\right)Dt},
			 \end{align}
		 where $ a_{nj}$ is a constant.
			  From \eqref{eq:59}, we have
			  \begin{align}\label{eq:84}
			 \sum_{n=0}^\infty\sum_{j=0}^\infty a_{nj}  J_0\left(\frac{l_{j}}{a}\rho \right)
			  \cos\left(\frac{n\pi z}{b}\right)=f(\rho,\phi, z).
			  \end{align}
	To find $a_{nj}$, we multiply each side of \eqref{eq:84} by $J_0\left(\frac{l_{j}}{a}\rho \right)
	\cos\left(\frac{n\pi z}{b}\right)$, taking integrals with respect to  ${\rho}$ and ${z}$, respectively, and use the following orthogonality relations \cite{Lam:13:Lec}
	\begin{align}
	\int_0^a J_0\left(\frac{l_{j}}{a}\rho \right) J_0\left(\frac{l_{\tilde{j}}}{a}\rho \right) \rho \diff \rho=
	\begin{cases}
	\frac{a^2}{2},\qquad &\text{for }j=\tilde{j}=0  \\
	\frac{a^2 J_0^2\left(l_{j} \right) }{2}\delta_{j\tilde{j}},&\text{for }j,\tilde{j}>0
	\end{cases}
	\end{align}
		\begin{align}
		\int_{-\frac{b}{2}}^{\frac{b}{2}} \cos\left(\frac{n\pi}{b}z\right)\cos\left(\frac{\tilde{n}\pi}{b}z\right)\diff z=
		\begin{cases}
		b,& \text{for } n=\tilde{n}=0\\
		\frac{b}{2}\delta_{n\tilde{n}},& \text{for } n,\tilde{n}>0\\	
		\end{cases}
		\end{align}
	where $\delta_{x\tilde{x}}=1$ when $x=\tilde{x}$ and  $\delta_{x\tilde{x}}=0$ when $x\neq\tilde{x}$. Then, changing the variables of the integrals, i.e., $\rho\rightarrow \tilde{\rho}$, $\phi\rightarrow \tilde{\phi}$, and $z\rightarrow \tilde{z}$, we obtain
	\begin{align} \label{eq:71}
	&a_{00}=\int_{-\frac{b}{2}}^{\frac{b}{2}}\int_0^a\frac{2}{ a^2b}   \tilde{\rho} f(\tilde{\rho},\tilde{\phi}, \tilde{z})  \diff \tilde{\rho} \diff \tilde{z},\\\label{eq:72}
	&a_{0j}=\int_{-\frac{b}{2}}^{\frac{b}{2}}\int_0^a \frac{2}{ a^2 J_0^2\left(l_{j} \right) b} J_0\left(\frac{l_{j}}{a}\tilde{\rho} \right)  \tilde{\rho} f(\tilde{\rho},\tilde{\phi}, \tilde{z})  \diff \tilde{\rho}\diff \tilde{z},\\\label{eq:73}
&a_{n0}= \int_{-\frac{b}{2}}^{\frac{b}{2}}   \int_0^a \frac{4}{ a^2b  } \cos\left(\frac{n\pi}{b}\tilde{z}\right) \tilde{\rho} f(\tilde{\rho},\tilde{\phi}, \tilde{z}) \diff \tilde{\rho} \diff \tilde{z},\\\label{eq:74}
		&a_{nj}= \\\nonumber
		&\int_{-\frac{b}{2}}^{\frac{b}{2}} \int_0^a   \frac{4}{ a^2 J_0^2\left(l_{j} \right)b } \cos\left(\frac{n\pi}{b}\tilde{z}\right)J_0\left(\frac{l_{j}}{a}\tilde{\rho} \right)  \tilde{\rho} f(\tilde{\rho},\tilde{\phi}, \tilde{z}) \diff \tilde{\rho} \diff \tilde{z}.
	\end{align}
Substituting \eqref{eq:71}, \eqref{eq:72}, \eqref{eq:73}, and \eqref{eq:74} into \eqref{eq:70} and setting $t_0=t$ and $t$ in this proof equal to $\Delta t$, we obtain \eqref{eq:55}.

	\bibliographystyle{IEEEtran}
	\bibliography{IEEEabrv,MolecularBib}

\end{document}